\documentclass[journal]{IEEEtran}
\usepackage{graphicx}
\usepackage{amsmath}
\usepackage{amsthm}
\usepackage{amssymb}
\usepackage{import}
\usepackage{color}
\usepackage{array}
\usepackage{diagbox}
\usepackage{multicol}
\usepackage{multirow}
\usepackage{placeins}
\usepackage[linesnumbered,ruled]{algorithm2e}
\usepackage{hyperref}
\usepackage{balance}
\usepackage[normalem]{ ulem }
\usepackage{soul}
\usepackage{mathtools, nccmath}
\usepackage{caption}
\usepackage{subcaption}
\title{Estimation of road traffic state at a multi-lanes controlled junction}

\author{
\IEEEauthorblockN{Cyril Nguyen Van Phu\IEEEauthorrefmark{1} and Nadir Farhi} \\
\IEEEauthorblockA{COSYS-GRETTIA, Univ Gustave Eiffel, IFSTTAR, F-77454 Marne-la-Vall\'ee, France} \\
\IEEEauthorblockA{\IEEEauthorrefmark{1} corresponding author}
\thanks{
Corresponding author : Cyril Nguyen Van Phu (email: cyril.nguyen-van-phu@univ-eiffel.fr)}
}

\newtheorem{prop}{Proposition}

\DeclareMathOperator{\tmid}{mid}
\begin{document}
\maketitle
\begin{abstract}
 We present in this paper a method for the estimation of traffic state at road junctions controlled with traffic lights.
 We assume mixed traffic where a proportion of vehicles are equipped with communication resources.
 The estimation of road traffic state uses information given by communicating vehicles.
 The method we propose is built upon a previously published method which was applied to estimate the traffic in the case where roads are composed of two lanes.
 In this paper, we consider the case where roads are composed of three lanes and we show that this solution can address the general case,
 where roads are composed of any number of lanes.
 We assume the geometry of the road junction is known, as well as its connections between incoming and outgoing lanes and roads.
 Using the location data provided by the communicating vehicles, first, we estimate some primary parameters including
the penetration ratio of the probe vehicles,
 as well as the arrival rates of vehicles (equipped and non-equipped) per lane by introducing the assignment onto the lanes.
 Second, we give estimations of the queue length of the $3$-lanes road, without and with the additional information provided by the location of the 
 communicating vehicles in the queue.
 We illustrate and discuss the proposed model with numerical simulations.
\end{abstract}

\begin{IEEEkeywords}
Intelligent transportation systems, Road transportation, Queuing systems.
\end{IEEEkeywords}

% ============================
\section{Introduction and state of the art}
\subsection{Introduction}
\label{sec:intro}
Limited capacities of roads and junctions, combined with traffic demand, determine the road traffic conditions experimented by the users in daily life.
Road traffic can then be modeled by shared resources systems such as queuing systems, as it has been done for example in~\cite{VARAIYA2013177},
where the max pressure algorithm
adapted to road traffic is presented.
In order to improve road traffic conditions experimented by the users, there is the possibility to control road traffic by guiding the users in the network,
or by controlling the traffic lights to reduce the delays.
However, controlling the road traffic needs an information on the state of the traffic.
In particular, concerning the road traffic state we are interested in the queue lengths at the junctions.

Nowadays, road traffic can be probed from the inside with communicating vehicles equipped with localization capabilities.
A probe vehicle is a vehicle which uses wireless communication to send information to another vehicle (vehicle to vehicle V2V), the infrastructure (vehicle to infrastructure V2I)
or to any other device (V2X).
The data provided by these mobile sensors are quite different in their nature from the data provided by fixed sensors.
With probe vehicles, we get data concerning individual sample vehicles trajectories, rather than global information on traffic state at a fixed location.
This difference in the nature of the information furnished, raises the opportunity to develop new methods for queue length estimation.
For example, in \cite{COMERT2016502} a stochastic method has been proposed to evaluate road traffic parameters like the queue lengths with probe vehicles.
In \cite{9011732}, we have presented a method to estimate road traffic state at controlled junctions for the two-lanes roads which extended the method in \cite{COMERT2016502}.

In the present paper, we aim at generalizing the method published in~\cite{9011732} to roads composed of any number of lanes.
This includes estimating road traffic primary parameters such as the penetration ratio of equipped vehicles, the arrival rate of vehicles and
the queue lengths per lane at the junction. We manage to do these estimations by introducing the assignment of vehicles onto the lanes. % in section~\ref{sub-primary}.

The outline of the paper is as follows : after the introduction~\ref{sec:intro}, we give a brief state of the art concerning the estimation of road traffic in subsection~\ref{sub-state-art}.
In section~\ref{sec-traffic-state} we introduce the problem and the notations used in the paper as well as the main estimation model.
In subsection~\ref{sub-primary} we introduce the assignment onto the lanes in order to balance the queues as much as possible and derive
primary parameters such as the penetration ratio of equipped vehicles, the arrival rate of vehicles per lane, the total arrival rate of vehicles, the probability distribution of
the queue lengths per lane.
In subsection~\ref{sub-three-lanes}, we give three different probability distributions for the queue lengths on a three lanes road.
By computing the expectations for these probability distributions we estimate the queue lengths at the traffic light.
The method presented in the present paper is general for any number of lanes. We show in section~\ref{sub-n-lanes} that the estimation of the queue lengths of three lanes roads is sufficient to
address the general n-lanes roads case. In section~\ref{sec:numerical}, we perform some numerical simulations that we have conducted with Omnet++~\cite{Varga01theomnet++}, 
a discrete event simulator, and analyze the results.
We conclude in section~\ref{sec-conclusion}.

\subsection{State of the art}
\label{sub-state-art}
\subsubsection{Road traffic estimation with fixed sensors (historical approaches)}
Historically, the estimation of road traffic state was done using sensors placed at fixed locations, such as magnetic loops, 
piezoelectric sensors or video cameras~\cite{leduc2008road}.
Among these approaches with fixed sensors, some estimations of the queue length and of the delay of vehicles at a traffic light have been given
in~\cite{beckmann1955studies}~\cite{webster1958traffic}~\cite{mcneil1968solution}~\cite{miller1968capacity}~\cite{akccelik1980time}.
These papers give analytic formulations for the under-saturated and over-saturated (i.e. when the arrival flow exceeds the intersection capacity) cases.
Among the input-output class of methods, we also cite~\cite{RePEc:eee:transb:v:44:y:2010:i:1:p:120-135}
where a probabilistic model for the estimation of queue lengths 
at signalized junctions, which can capture spillback and gridlock phenomena, is presented.
In 2009, the authors of~\cite{LIU2009412} take advantage of shockwave traffic theory, combined with fixed detector and signal timings data input, in order to estimate queue lengths.

\subsubsection{Road traffic estimation with probe vehicles}
In this paragraph, we give an overview of the methods that have been developed upon the information given by these new mobile sensors, namely, probe vehicles.
In~\cite{GUO2019313}, the authors provide a timely survey on traffic information collection and state estimation methods published in the last decade,
which use the data provided by connected and automated vehicles (CAVs). 
They classify the different traffic observation methods which use CAVs data by 
distinguishing between deterministic and stochastic approaches. We use here the same outline.

Among the deterministic approaches, we cite~\cite{doi:10.1111/mice.12095} where the authors expose a method to estimate queue lengths with probe vehicles as the single source of information.
Position and instantaneous speed of probe vehicles are the input data of their method.
The latter relaxes some common assumptions made in the literature,
such as the knowledge of signal timings or arrival process distribution.
The shockwave theory based on first order traffic models is the key model used in their queue length estimation method.
In~\cite{JEFFBAN20111133}, Ban et al. use intersection travel times in order to estimate queue lengths and delays at junctions.
These intersection travel times are measured when probe vehicles cross virtual trip lines (VTL), located upstream and downstream relatively to the intersection.
The benefits of using intersection travel times are~: respect the privacy of the users, the flexibility in defining the virtual trip lines (as they are virtual locations), 
and the pliancy which enables other sensor inputs such as Bluetooth Mac address matching, and other travel times collection systems.

Among the stochastic approaches, the authors of~\cite{HAO2013513} define a vehicle index as ``the position of vehicles in the departure process of the cycle''.
That paper has proposed a method for estimating these vehicle indices which are described as a basic and primordial information that can be provided by
probe vehicles. For example, knowing the index of a vehicle gives its position in the queue.
With the intent of solving some privacy issues, their method relies only on intersection travel times as input data.
They derive the intersection travel times from the arrival time and departure time of probe vehicles into virtual areas (Virtual Trip Lines, VTL) respectively 
upstream and downstream the intersection.
They model the arrival process as a time-dependent Poisson process; and use a log-normal distribution to model the departure headways for every vehicle index.
In this framework, the authors use a Bayesian Network in order to estimate vehicle indices.

Vehicle indices are some basic information that can be used as input data for estimating queue length at junctions.
Indeed in 2014, Hao et al.~\cite{HAO2014185} have naturally pursued their work on vehicle indices with the estimation of queue lengths at intersections.
Using as input data the intersection travel times and the vehicle indices, as determined in their former 2013 paper,
they estimate queue lengths using a stochastic model based on Bayes theorem.

In 2009, Comert and Cetin~\cite{COMERT2009196} have proposed a method for the estimation of queue length using the data provided by probe vehicles.
They have assumed that probe vehicles indices are available as input data.
Assuming that the probability distribution of the queue length is given, they compute a conditional probability distribution of the queue length, knowing 
the locations of the probe vehicles in the queue. They show that for the 1-lane case, the location of the last probe vehicle in the queue is the only one needed.
The authors of~\cite{COMERT2009196} have also derived the variance of the estimator. 
Furthermore, numerical analysis are performed, where the arrival processes models and arrival processes intensities are varied.
This work ``appears to be the first attempt to formulate the problem of estimating the queue length from probe
vehicle data.''

In 2013, Comert~\cite{COMERT201359} derives queue length, last probe location and queue joining times probability distributions, with or without overflow queue (residual queue at the end of the red time).
Mean and variance for the queue length estimators are given.
It has been shown that the estimators depend on probe proportion, red duration and arrival process properties.

In 2016, Comert~\cite{COMERT2016502} goes a step further by studying the cases with unknown probe proportions and unknown arrival rates.
He gives analytical formulations for these primary parameters (proportions of probe vehicles and arrival rate), as well as various
queue length estimators with or without overflow queue.
Derivation of the estimation errors are also given, and numerical analysis performed with VISSIM microscopic simulator have been presented.

In 2017, Zheng and Liu~\cite{ZHENG2017347} estimated traffic volumes for low penetration ratio of equipped vehicles.
The method proposed uses as input data ``vehicle trajectories approaching to an intersection as well as traffic signal status''.
The trajectories of equipped vehicles are used to detect if a probe vehicle has stopped at the traffic light and its stopping position.
With these information,  the arrival rate is estimated and bounds for this arrival rate are given.
Zheng and Liu have used for their estimation a time dependent Poisson arrival process and the Expectation Maximization (EM) algorithm.
They have tested their method with data sets from an experiment where around 2800 probe vehicles were deployed in the city of Ann Arbor, and
from data provided by commercial navigation service in China.
Machine learning methods have also been used in order to predict traffic flow;
for example, Lv et al.~\cite{2015LvBigData} used a data driven machine learning method :
the stacked autoencoders model is combined with
the data provided by 15000 individual detectors deployed across California in order to predict traffic conditions.
In 2021, Zhao et al.~\cite{zhao_maximum_2021} have proposed a method based on maximum likelihood algorithm in order to estimate the penetration ratio of probe vehicles as well as the distribution of queue lengths.
The expectation-maximization algorithm is used to solve the problem formulated as maximum likelihood. The accuracy of the estimation results shows the relevance of their method.
We also refer to the work by Tang et al. published in 2021~\cite{tang_lane-based_2021} in which the authors propose a method for the estimation of queue lengths using license plate recognition detectors.
The authors introduce in their paper a lane based estimation of queue lengths which gives promising simulation and empirical results.
Although the work by Tang et al.~\cite{tang_lane-based_2021} is not using probe vehicles location as input data,
it is addressing the multi lanes case combined with shockwave or input output methods.
\subsection{Assumptions and main contributions}
\label{sec:contrib}
For the roads composed of many lanes, it is not accurate to use the model proposed by Comert~\cite{COMERT2009196}~\cite{COMERT201359}~\cite{COMERT2016502}.
This is because the shortest lanes would be equated to the longest lanes, which obviously is false, in particular in the case of unbalanced demands on the various lanes.
Controls of the traffic light which are lane based would be inaccurate by not differentiating between the lanes.

In the present paper, we generalize the method presented in~\cite{9011732} which estimates road traffic state at controlled junctions for the two-lanes roads, to the case
where roads are composed of any number of lanes.
These methods are based on existing works published in~\cite{COMERT2009196}~\cite{COMERT201359}~\cite{COMERT2016502} which were addressing the roads composed of one lane.
To our knowledge, the present work and our previous paper~\cite{9011732} are the first works to extend the papers by Comert and Cetin~\cite{COMERT2009196}~\cite{COMERT201359}~\cite{COMERT2016502} to the multi lanes case.
In addition to this generalization, we give methods to estimate the number of probe vehicles per lane as well as a new estimator for the penetration ratio of probe vehicles.
Based on these estimates, we also give a new analytical formulation for estimating the queue length per lane (Proposition 3).

The main assumptions of our work are the following :
\begin{itemize}
 \item a proportion $p$ of vehicles are equipped with wireless communication and localization capabilities, and we name them probe vehicles. The probe vehicles send their GPS localization to a road side unit using wireless communication.
The penetration ratio $p$ of probe vehicles is considered as a variable in the present paper.
 \item the traffic demand is low or moderate such that we can assume that the arrivals are following a Poisson process. Hence it is supposed that there is no overflow queue : the queues are cleared at the beginning of the red time.
 \item it is assumed that in average the queues tend to balance among the different lanes under the constraints of the vehicle assignments onto the lanes, with respect to their destinations.
 \item the GPS localization of vehicles is not accurate enough to distinguish on which lane a probe vehicle is located. This assumption relies on the
 fact that the accuracy of GPS localization is such that
 we can assign a vehicle with 95\% probability within a radius of $7.8 m$~\cite{gps.gov} although the width of a lane is $3.7m$ in the USA.
\end{itemize}

\section {Road traffic state estimation}
\label{sec-traffic-state}
% ============================
The notations we use in the present paper are described in TABLE~\ref{tab:notations}; see also Fig.~\ref{fig:queue_3_lanes}.
\begin{figure}[htbp]
\begin{center}
\includegraphics[width=0.9\linewidth, keepaspectratio]{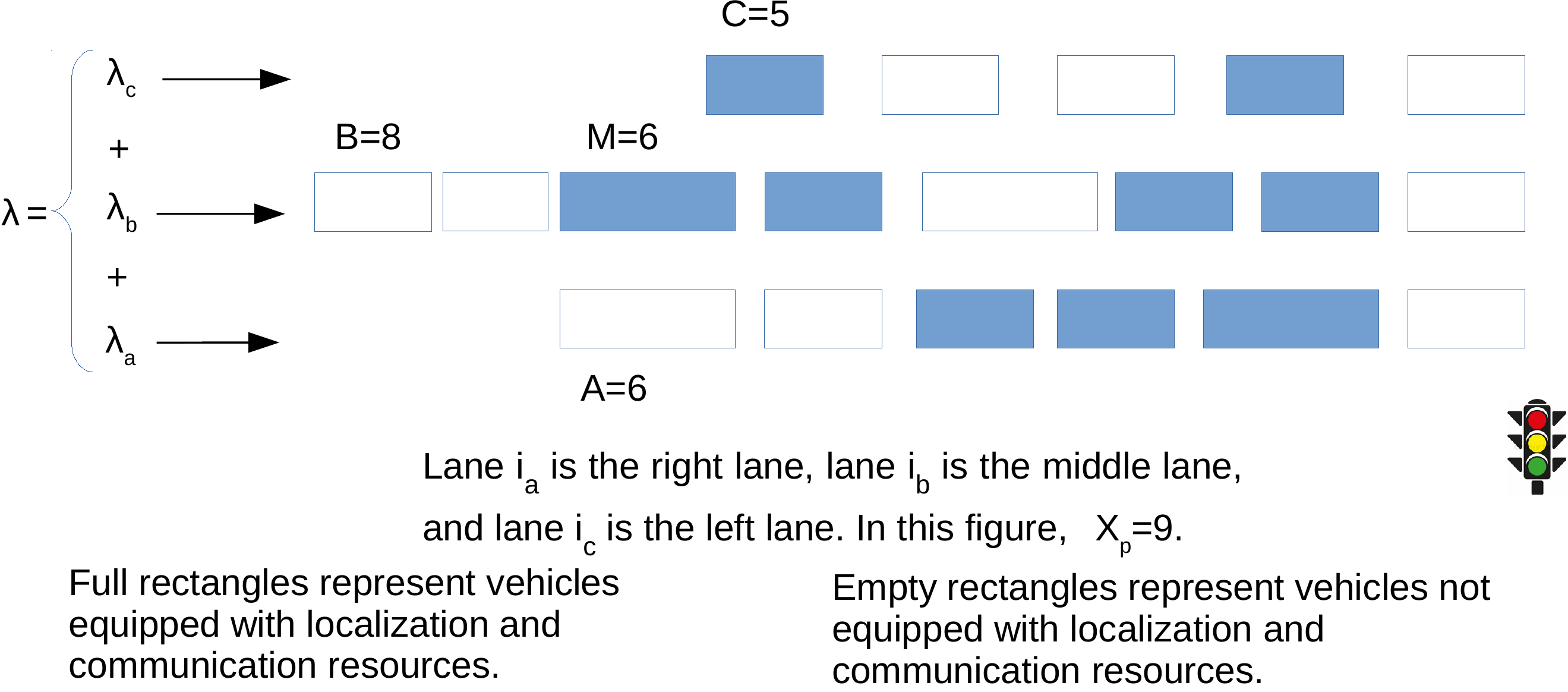}
\caption{Vehicles queue at a controlled junction, in a three-lanes incoming road.}
\label{fig:queue_3_lanes}
\end{center}
\end{figure}
\begin{table}[htbp]
  \begin{tabular}{|l|p{4.5cm}|}
    \hline
    Name & Definition \\
    \hline
    $R$ & the total red time in one cycle \\
    \hline
    $r_a(t)$, $r_b(t)$, $r_c(t)$  & the time since the beginning of the red phase for respectively lane $i_a$, $i_b$, or $i_c$ (it is $0$ if we are not in red phase at time $t$), $0\leq r_a(t), r_b(t), r_c(t)\leq R$. \\
    \hline
    $\lambda_a, \lambda_b$, $\lambda_c$ & the average arrival rate in vehicles/second respectively on lane $i_a$, $i_b$ and $i_c$.\\
    \hline
    $\lambda=\lambda_a+\lambda_b+\lambda_c$ & the total arrival rate for the incoming link in vehicles/second.\\
    \hline
    $X(t)$ & the number of vehicles queuing on all the lanes, of the considered link at time~$t$. \\
    \hline    
    $X_{p}(t)$ & the number of communicating vehicles queuing on all the lanes of the considered link at time~$t$. \\
    \hline    
    $Y(t)$ & the number of vehicles (queuing and not queuing) of the considered link at time $t$. \\
    \hline    
    $Y_{p}(t)$ & the number of communicating vehicles (queuing and not queuing) on all the lanes of the considered link at time $t$. \\    
    \hline    
    $p$, $0 \leq p \leq 1$ & the penetration ratio of probe vehicles. \\
    \hline    
    $A$, $B$, $C$ & the number of vehicles in the queue at time $t$ and respectively lane $i_a$, $i_b$, or $i_c$. In this paper, $A$, $B$, $C$ are assumed to be random variables. \\
    \hline    
    ($A_p$, $\hat{a_p}$), ($B_p$, $\hat{b_p}$), ($C_p$, $\hat{c_p}$) & the number of probe vehicles in the queue, and its estimate, at time $t$ and respectively at lane $i_a$, $i_b$, or $i_c$. In this paper, $A_p$, $B_p$, $C_p$ are assumed to be random variables. \\
    \hline    
    $M$ & the location (in number of vehicles) of the last probe in the queue, namely the last connected vehicle, at time $t$. $M$ is assumed to be a random variable, taking value $m$.\\
    \hline    
    $L$ & the incoming lane $i_a$, $i_b$, or $i_c$, of the last connected vehicle, at time $t$. $L$ is assumed to be a random variable.\\
    \hline    
    $O_k$ & $O_k=i$ if a vehicle $k$ is located on an incoming lane $i={i_a,i_b,i_c}$. $O_k$ is considered a random variable.\\
    \hline
    $D_k$ & $D_k=j$ if a vehicle $k$ is located on an outgoing road $j={1,2,..,d}$. $D_k$ is considered a random variable.\\
    \hline    
    $W_{ij}$ & $W_{ij}$ is a random variable. $W_{ij}=1$ if a vehicle $k$ comes from the origin lane $i$ and goes to a destination road $j$, such that $O_k=i$ and $D_k=j$; $W_{ij}=0$ otherwise. We have the following notation : $P(W_{ij}=1)=w_{ij}$.\\
    \hline    
    $W=(w_{ij})_{i=i_a, i_b, i_c;j=1..d}$ & W is the matrix with three lines and $d$ columns which represents the vehicles assignment from the incoming lanes to the outgoing roads.\\
    \hline
    $q_{sat}$ & the saturation rate (exit rate) of roads outgoing from the junction \\
    \hline    
    $t_e^{k,j}$ & the exit time of the junction of a probe vehicle $k$ on an outgoing road $j$\\
    \hline    
    $\pi(k,\mu_i)=\mu_i^k e^{-\mu_i}/k!$ & the Poisson probability mass function of parameter $\mu_i$\\
    \hline
  \end{tabular}
  \caption{Notations}
  \label{tab:notations}
\end{table}
We consider a road junction composed of a number of incoming roads, one of them having three lanes, and $d$ outgoing roads of any number of lanes, controlled by a traffic light.
The vehicles come from an origin lane $i\in O=\{i;i=i_a,i_b,i_c\}$ (where $i_a$, $i_b$, $i_c$ denote the three origin lanes)
of the 3-lanes incoming road; 
and go to a destination road $j\in D:=\{j;j=1..d\}$.
We label the incoming vehicles by index $k\in \mathbb{N}^*$ without taking into account their arrival order.
Then, we introduce the two families $O_k$ and $D_k, k\in\mathbb{N}^*$  of random variables.
$O_k$ takes its values in $O$, such that $O_k=i$ if vehicle $k$ comes from origin lane $i$.
$D_k$ takes its values in $D$, such that $D_k=j$ if vehicle $k$ goes to destination road $j$.
We assume that the probability that a vehicle $k$ comes from incoming lane $i\in O$ and 
goes to outgoing road $j$, is the same for all the vehicles $k\in \mathbb{N}^*$.
Hence, we introduce the family of boolean random variables $W_{ij}, i\in O, j\in D$ such that
$W_{ij}=1$ if a vehicle $k$ comes from the origin lane $O_k=i$ and goes to destination road $D_k=j$; and $0$ otherwise.
We will denote $w_{ij}:=P(W_{ij}=1)$,
and $W:=(w_{ij})_{i\in O;j\in D}$ the matrix with three lines and $d$ columns which represents the vehicles assignment from the incoming lanes to the outgoing roads,
and the assignment weight given to each couple (incoming lane, destination road).
\subsection{Primary parameters}
\label{sub-primary}
In this section, we determine the penetration ratio $p$ of equipped vehicles, the total arrival rate $\lambda$ of vehicles, the matrix $W$,
and the arrival rates $\lambda_i$ per incoming lane.

\subsubsection{Estimation model for the penetration ratio p}
We will now estimate the penetration ratio $p$ of probe vehicles.
In \cite{9011732}, we have introduced the ratio $\kappa=\min{\mu_i}/\max{\mu_i}$ and demonstrated that for the two lanes case and under the conditions of taking respectively the maximum and minimum queue lengths measurements as :
$l_p+1/p-1$, $\kappa(l_p+1/p-1)$, our estimator $\hat{p}=\frac{X_p/(1+\kappa)-1}{m-1}$ is unbiased. 
However, these two queue lengths estimations are not given usually as measurements and
we can not take $\kappa$ as an average value because it will introduce some bias.
This is why we propose another estimator for the penetration ratio $p$ of probe vehicles which can be applied more easily in real life situations.
The number of probe vehicles in an incoming road $n$ is denoted by $X_p(n)$.
Hence, the total number of probe vehicles in the queues of all incoming roads is $\sum_n X_p(n)$ where the index $n$ of the sum represents an incoming road.
We denote $t_e^{k,j}$ the exit time of a probe vehicle $k$ on an outgoing road $j$ during the green time. 
Considering one outgoing road $j$, we can say that the last probe vehicle going out to this road is located at $q_{sat}\times\max_k\{t_e^{k,j}\}$ position (in number of vehicles)
where $q_{sat}$ is the saturation rate (exit rate) of vehicles per unit of time.
If we consider the vehicles waiting in the queue, we can virtually rearrange their order such that they are placed in the order they are going out of the junction.
We notice that changing the order of vehicles does not change their total number.
We write the following formula for the estimation of $p$ where the index of the sum $n$ represents an incoming road and the index of the sum $j$ represents an outgoing road.
\begin{equation}
\label{eq:p}
 \hat{p}=\frac{\sum_n X_p(n)}{q_{sat}\sum_j \max_k\{t_e^{k,j}\}}
\end{equation}
We notice that the numerator of equation~(\ref{eq:p}) is computed during the red time and the denominator during the green time.
\subsubsection{Estimation model for the total arrival rate $\lambda$}
We also propose to estimate $\lambda$ as follows :
\begin{equation}
 \hat{\lambda}=\frac{Y_p(R)-Y_p(0)}{pR}
\end{equation}
$Y_p$ represents the number of probe vehicles on the considered incoming road and $R$ is the total red time.
\subsubsection{Estimation model for the matrix W}
The matrix $W$ is used to calibrate the assignment model, by equilibrating the queue lengths among the three lanes $i_a$, $i_b$, $i_c$.
The first step is to determine the matrix $W$.
The probability that a vehicle $k$ is coming from lane $i$ is equal to $w_i=\sum_j w_{ij}$.
We have $0\leq w_i\leq1$ and $\sum_i w_i=1$.
The objective is to find the optimal assignment matrix $W$ which equilibrates the ratios of inflows over the three lanes of the incoming road.
The ideal case is :
\begin{equation}
 \forall i, \sum_j w_{ij}=1/3
\end{equation}
Therefore, in order to equilibrate the inflows over the three incoming lanes we propose to minimize the difference $(\sum_j w_{ij}-1/3)$ for all $i=1,2,3$.
We denote $\bar{W}:=(w_{11},..., w_{1d}, w_{21},...,w_{2d},w_{31},...,w_{3d})$
and $v=(1/3, 1/3, 1/3)$.
We denote $\bar{\lambda_j}$ the proportion of the arrival rate $\lambda$ of vehicles going to the outgoing road $j$.
The $\tilde{w}_{ij}=0$ represent the information given by the topology of the road junction, 
specifically the incoming lanes and outgoing roads which are not connected.
We have the following constraints :
\begin{align}
\sum_{ij}{w_{ij}}=1  \label{cont1} \\
\sum_i w_{ij}  = \rho_j := \bar{\lambda_j}/\lambda,  \forall j \label{cont2} \\
\tilde{w}_{ij} = 0 \label{cont3} \\
0\leq w_{ij}\leq 1 \label{cont4} 
\end{align}
We note that, since the turn ratios $\rho_j$ are fixed such that $\sum_j \rho_j = 1$, then the constraint~(\ref{cont1}) is automatically satisfied.
The constraints~(\ref{cont2}) and~(\ref{cont3}) are linear. We can write them as follows~:
\begin{equation}
\label{eq:constraints2}
 \bar{A}\bar{W}=\bar{b}
\end{equation}
where
$$\bar{A} = \begin{pmatrix}                        
		I_d & I_d & I_d \\
		    & B   &
	    \end{pmatrix},$$
and $\bar{b}=(\rho_1,..., \rho_d,0, ..., 0)$, with $I_d$ the $d\times d$ identity matrix, and $B$ the matrix satisfying $B \bar{W} = \tilde{w}$.

For the criterion, we define the $3\times(3d)$ matrix $L$ such that~:

\begin{equation}
 \left(w_1, w_2, w_3\right)=\left( \sum_j w_{1j}, \sum_j w_{2j},\sum_j w_{3j} \right)= L \bar{W}.
\end{equation}
We write the following objective function :
\begin{align}
\label{eq:w}
  \min_{\bar{W}} (L\bar{W}-v)'(L\bar{W}-v)
\end{align}
This can be written as :
\begin{align}
\label{eq:w1}
   \min_{\bar{W}} \bar{W}'Q\bar{W} - 2(v'L)\bar{W}
\end{align}
where $Q=L'L$.
Finally, the minimization problem can be written as :
\begin{align}
\label{eq:w2}
  \min_{\bar{W}} \bar{W}'Q\bar{W} - 2(v'L)\bar{W} \nonumber\\ 
\bar{A}\bar{W}=\bar{b}\\
0\leq w_{ij}\leq 1 \nonumber
\end{align}
The problem~(\ref{eq:w2}) is convex since the criterion is quadratic and the constraints are linear.
Therefore, first order conditions of optimization are necessary and sufficient to solve this problem.
Practically, we use an off-the-shelf optimization library provided with Octave software~\cite{octave} in order to solve this problem 
and derive $\bar{W}$ (and then the matrix $W$), given the turn ratios as input.
\subsubsection{Estimation of the arrival rates per lane $\lambda_i$}
The arrival rates $\lambda_i$ per incoming lane $i$ in one cycle are given as follows~:
\begin{align}
 \lambda_{i} &=  \lambda \sum_j w_{ij} = \lambda w_i
 \label{eq:lambda_i}
\end{align}
$(\lambda_i)_{i=i_a, i_b, i_c}$ represents the arrival rate of vehicles in the queue, respectively on lane $i=i_a, i_b, i_c$.
This will enable us to determine the probability distribution of queue lengths per lane in Proposition~\ref{prop1}.

\subsection{Traffic state estimation for roads composed of three lanes}
\label{sub-three-lanes}
In this section, we propose three probability distributions for the queue lengths.
We will use the expectation of these probability distributions as estimators.
Proposition~\ref{prop1} gives the probability distribution without the information given by the probe vehicles but using the matrix W, especially the arrival rates per lane $\lambda_i$.
Proposition~\ref{prop2} will refine Proposition~\ref{prop1} by adding the information given by the probe vehicles. It is an extension of a previously published result~\cite{9011732}.
In Proposition~\ref{prop3}, we use the estimation of probe vehicles per lane to give another estimator.

We denote $\mu_i$ the stock of vehicles waiting in the queue of lane $i$ at time $t$ :
\begin{equation}
 \mu_i=\lambda_i r_i(t)\nonumber
\end{equation}
We denote $\pi(k,\mu_i)$ the Poisson probability mass function of parameter $\mu_i$.
$$\pi(k,\mu_i)=\frac{\mu_i^k e^{-\mu_i}}{k!}, \forall i$$
Under the assumption of an arrival process of vehicles represented by a Poisson process of rate 
$\lambda=\lambda_a+\lambda_b+\lambda_c$, we have the following proposition :
\begin{prop}
  \label{prop1}
  \begin{multline}
    P(A=a,B=b,C=c)=\pi(a,\mu_a)\pi(b,\mu_b)\pi(c,\mu_c) \nonumber\\    
  \end{multline}
\end{prop}
\begin{proof}
The Poisson process of rate $\lambda$ is subdivided into three independent Poisson processes of rate $\lambda_a$, $\lambda_b$, $\lambda_c$ 
with probability respectively $\lambda_a/\lambda$, $\lambda_b/\lambda$, $\lambda_c/\lambda$.
By subdividing the main Poisson process of parameter $\lambda$ with these probabilities, we get
three Poisson processes of parameter $\lambda(\lambda_i/\lambda)=\lambda_i$ for
each lane $i$. Furthermore, the three Poisson processes are independent.
Concerning the subdividing of a Poisson process, we cite reference~\cite{gallager2013stochastic}.
Because of the stationary increment property of a Poisson process,
the expected number of vehicles queuing on lane $i$ at time $r_i(t)$ is $\lambda_i r_i(t)=\mu_i$.
\end{proof}
Knowing the location of the last probe vehicle into the queue and the total number of probe vehicles in the queue,
we can refine the estimation of the queue lengths :
\vspace{0.10cm}
\begin{prop}
\label{prop2}
\hspace{2em}
\\
If $m \leq\max(a,b,c)$ and $x_p \leq a+b+c$, then :
      \begin{multline}
	P(A=a,B=b,C=c|M=m,X_p=x_p)=\\
	\frac{\sigma_{a,b,c}(1-p)^{a+b+c}P(A=a,B=b,C=c)}
	{\sum \limits_{\substack{j,k,l\geq0\\\text{subject to}\\m \leq \max(j,k,l)\\x_p\leq j+k+l}}^{}
	\sigma_{j,k,l}(1-p)^{j+k+l} P(A=j,B=k,C=l)}. \nonumber    
      \end{multline}
    Otherwise, \\ $P(A=a,B=b,C=c|M=m,X_p=x_p)=0$.
\\
where we define :
$$\sigma_{j,k,l}:=\binom{m-1+\min(m,j,k,l)+\tmid(j,k,l)}{x_p-1}$$
and $\tmid(j,k,l):=(j+k+l)-\min(j,k,l)-\max(j,k,l)$
\end{prop}
\begin{proof}
 The proof given in~\cite{9011732} for the Proposition~\ref{prop2} can be directly applied to our present case, with minor changes.
\end{proof}

We now propose to compute the estimation of the queues by adding the information of matrix $W$ which summarizes the knowledge of the destinations of the probe vehicles at the road junction.
Hence, we propose to compute for lane $i_a$ $$p_A(a)=P(A=a|M=m,A_p=\hat{a_p})$$ given the estimated number $A_p=\hat{a_p}$ of probe vehicles on lane $i_a$ and the location of the last probe in the queues.
First, we will give estimations of the number of probe vehicles per lane $A_p=\hat{a_p}$ in order to compute $p_A(a)$.
We have :
$$A_p = \sum_k \textbf{1}_{a_k},$$
where $$\textbf{1}_{a_k} := \begin{cases}
                              1 & \text{ if vehicle $k$ comes from lane $i_a$} \\
                              0 & \text{otherwise}
			    \end{cases}$$
Then,
\begin{align}
   \mathbb E(A_p) & = \sum_k \mathbb E (\textbf{1}_{a_k}) \nonumber \\
		  & = \sum_k P(O_k =i_a) \nonumber \\
		  & = \sum_{k,j} P(O_k =i_a | D_k = j) P(D_k = j) \nonumber
\end{align}
where $D_k$ is observed. Then,
$P(D_k = j_k) = 1$ for the vehicle $k$ after it has crossed the junction and has gone through an outgoing road $j_k$
and $P(D_k = j) = 0$ otherwise, for $j\neq j_k$.

We have :
\begin{equation}
\label{eq:probes}
 \mathbb E(A_p)= \sum_{k} P(O_k=i_a|D_k=j_k)
\end{equation}
\begin{equation}
\nonumber
 \mathbb E(B_p)= \sum_{k} P(O_k=i_b|D_k=j_k)
\end{equation}
\begin{equation}
\nonumber
 \mathbb E(C_p)= \sum_{k} P(O_k=i_c|D_k=j_k)
\end{equation}
We propose to estimate the number of probe vehicles on lanes $i_a$, $i_b$, $i_c$ by respectively :
$\hat{a_p}=[\mathbb E(A_p)]$, $\hat{b_p}=[\mathbb E(B_p)]$, $\hat{c_p}=[\mathbb E(C_p)]$, where $[x]$ denotes rounding $x$ to the nearest integer.

Now we propose to compute $p_A(a)$ the queue length probability distribution on lane $i_a$ with the information provided by the estimation of probe vehicles on lane $i_a$ ($\hat{a_p}$) and
the location of the last probe vehicle on all the lanes ($m$).
The following proposition is also true for the other lanes by just inverting the lane $i_a$ with $i_b$ or $i_c$.
\small
\begin{prop}
\label{prop3}
\begin{align}
\nonumber
&\text{if } a_p\geq 1, m\geq1, a\geq a_p \text{, then } p_A(a)= \\\nonumber
&\frac{\left(\lambda_a\binom{m-1}{\hat{a_p}-1}+\lambda_b S_{\mu_b}^{m,\hat{a_p}} \binom{a}{\hat{a_p}}+\lambda_c S_{\mu_c}^{m,\hat{a_p}} \binom{a}{\hat{a_p}}\right)
(1-p)^{a}\pi(a,\mu_a)}
{\sum\limits_{n\geq a_p}\left(\lambda_a\binom{m-1}{\hat{a_p}-1}+\lambda_b S_{\mu_b}^{m,\hat{a_p}} \binom{n}{\hat{a_p}}+\lambda_c S_{\mu_c}^{m,\hat{a_p}} \binom{n}{\hat{a_p}} \right)
(1-p)^{n}\pi(n,\mu_a)}\\\nonumber
&\text{if } a_p=0 \text{, then }p_A(a)=\frac{(1-p)^{a}\pi(a,\mu_a)}{\sum_{n\geq0}(1-p)^{n}\pi(n,\mu_a)}\\\nonumber
&\text{if } a<a_p \text{ or } m<a_p \text{, then }p_A(a)=0\nonumber
\end{align}
where we define :
$$S_{\mu}^{m,\nu}:= \sum_{j\leq k,k\geq \max(m,\nu)} \binom{m-1}{j-1}p^{j}(1-p)^{k-j}\pi(k,\mu)$$
\end{prop}
\normalsize
\begin{proof}
  By Bayes' theorem, we have :
\begin{equation}
p_A(a)=\frac{P(A=a,M=m,A_p=\hat{a_p})}{P(M=m,A_p=\hat{a_p})}
\end{equation}
The numerator can be written :

\begin{multline}
\label{eq:1}
 P(A=a,M=m,A_p=\hat{a_p})=\\
 P(M=m|A_p=\hat{a_p},A=a)P(A=a,A_p=\hat{a_p})
\end{multline}
Bayes theorem implies that :
\begin{equation}
P(A=a,A_p=\hat{a_p})=P(A_p=\hat{a_p}|A=a)P(A=a)
\end{equation}
So, we can write the second term of equation (\ref{eq:1}) as :
\begin{equation}
P(A=a,A_p=\hat{a_p})=\binom{a}{\hat{a_p}}p^{\hat{a_p}}(1-p)^{a-\hat{a_p}}P(A=a)
\end{equation}
Concerning the first term of the product in equation (\ref{eq:1}), we use the marginal distribution on the random variable $L$, and we can write :
\begin{multline}
 P(M=m|A_p=\hat{a_p},A=a)=\\
 P(M=m,L=i_a|A_p=\hat{a_p},A=a)+\\
 P(M=m,L=i_b|A_p=\hat{a_p},A=a)+\\
 P(M=m,L=i_c|A_p=\hat{a_p},A=a)
  \label{eq:2}
\end{multline}
Let us detail the first term of equation (\ref{eq:2}), in the case where $i=i_a$.
Concerning $P(M=m,L=i_a|A_p=\hat{a_p},A=a)$, we recall the arguments given in \cite{COMERT2009196} which are :
``The sample space for
the experiment is the possible combinations of choosing $\hat{a_p}$ probe
vehicles from $a$ vehicles, which is equal to $\binom{a}{\hat{a_p}}$.
The number of elements in the event space is equal to the number of possible placements of the remaining probe vehicles,
other than the one at position $m$, into the previous slots since $m$ is fixed.'', which is 
$\binom{m-1}{\hat{a_p}-1}$.
As the last probe should be on lane $i_a$ with probability $\lambda_a/\lambda$, we write :
\begin{equation}
 P(M=m,L=i_a|A_p=\hat{a_p},A=a)=\frac{\lambda_a}{\lambda}\frac{\binom{m-1}{\hat{a_p}-1}}{\binom{a}{\hat{a_p}}}
\end{equation}
In addition, we have :
\begin{multline}
  P(M=m,L=i_b|A_p=\hat{a_p},A=a)=\\
\frac{P(M=m,B \geq \max(m,\hat{a_p}),L=i_b|A_p=\hat{a_p},A=a)}{P(B \geq \max(m,\hat{a_p})|M=m,L=i_b,A_p=\hat{a_p},A=a)}
\end{multline}
Since $P(B \geq \max(m,\hat{a_p})|M=m,L=i_b,A_p=\hat{a_p},A=a)=1$, we write :
\begin{multline}
  P(M=m,L=i_b|A_p=\hat{a_p},A=a)=\\
P(M=m,B \geq \max(m,\hat{a_p}),L=i_b|A_p=\hat{a_p},A=a)
\end{multline}
We also have :
\begin{multline}
P(M=m,B \geq \max(m,\hat{a_p}),L=i_b)=\\
\frac{  P(M=m,B \geq \max(m,\hat{a_p}),L=i_b|A_p=\hat{a_p},A=a)}{P(A_p=\hat{a_p},A=a|M=m,B \geq \max(m,\hat{a_p}),L=i_b)}\\
\times P(A_p=\hat{a_p},A=a)
\end{multline}
%As the Poisson processes on lanes A, B, and C are independent, we can write :
and because $P(A_p=\hat{a_p},A=a|M=m,B \geq \max(m,\hat{a_p}),L=i_b)=P(A_p=\hat{a_p},A=a)$ we can write~:
\begin{multline}
  P(M=m,L=i_b|A_p=\hat{a_p},A=a)=\\
P(M=m,B \geq \max(m,\hat{a_p}),L=i_b)
\end{multline}
By computing the marginal distribution probability on $B_p$, we have :
\begin{multline}
 P(M=m,B \geq \max(m,\hat{a_p}),L=i_b)= \\
 \sum_{b_p,b\geq \max(m,\hat{a_p})} P(M=m, L=i_b|B_p=\hat{b_p},B=b)\times\\
 P(B_p=\hat{b_p},B=b)
\end{multline}
\begin{multline}
 P(M=m,B \geq max(m,\hat{a_p}),L=i_b)= \\
 \sum\limits_{b_p\leq b,b\geq \max(m,\hat{a_p})} \frac{\lambda_b}{\lambda}\frac{\binom{m-1}{\hat{b_p}-1}}{\binom{b}{\hat{b_p}}}\binom{b}{\hat{b_p}}p^{\hat{b_p}}(1-p)^{b-\hat{b_p}}P(B=b)
\end{multline}
\begin{multline}
 P(M=m,B \geq max(m,\hat{a_p}),L=i_b)=  \\
 \sum_{b_p\leq b,b\geq \max(m,\hat{a_p})} \frac{\lambda_b}{\lambda}\binom{m-1}{\hat{b_p}-1}p^{\hat{b_p}}(1-p)^{b-\hat{b_p}}P(B=b)
\end{multline}
We define the variable  $S_{\mu}^{m,\nu}$ :
$$S_{\mu}^{m,\nu}:= \sum_{j\leq k,k\geq \max(m,\nu)} \binom{m-1}{j-1}p^{j}(1-p)^{k-j}\pi(k,\mu)$$
Then, we can write :
\begin{equation}
 P(M=m,B \geq max(m,\hat{a_p}),L=i_b)= \frac{\lambda_b}{\lambda}S_{\mu_b}^{m,\hat{a_p}}
\end{equation}
Similarly,
\begin{equation}
 P(M=m,C\geq max(m,\hat{a_p}),L=i_c)= \frac{\lambda_c}{\lambda}S_{\mu_c}^{m,\hat{a_p}}
\end{equation}
Finally,
\begin{multline}
 P(M=m|A_p=\hat{a_p},A=a)=\\
 \frac{\lambda_a}{\lambda}\frac{\binom{m-1}{\hat{a_p}-1}}{\binom{a}{\hat{a_p}}}+\frac{\lambda_b S_{\mu_b}^{m,\hat{a_p}}}{\lambda}+\frac{\lambda_c S_{\mu_c}^{m,\hat{a_p}}}{\lambda}
\end{multline}
Hence, the numerator can be written as :
\begin{multline}
\label{eq:1:numerator}
 P(A=a,M=m,A_p=\hat{a_p})=\\
\left(\frac{\lambda_a}{\lambda}\frac{\binom{m-1}{\hat{a_p}-1}}{\binom{a}{\hat{a_p}}}+\frac{\lambda_b S_{\mu_b}^{m,\hat{a_p}}}{\lambda}+\frac{\lambda_c S_{\mu_c}^{m,\hat{a_p}}}{\lambda}\right)\times\\
  \binom{a}{\hat{a_p}}p^{\hat{a_p}}(1-p)^{a-\hat{a_p}}P(A=a)
\end{multline}
Using the marginal distribution and some simplifications, we can write in conclusion that :
% \footnotesize
% \begin{multline}
% \nonumber
% p_A(a)=\\
% \frac{\left(\frac{\lambda_a}{\lambda}\frac{\binom{m-1}{\hat{a_p}-1}}{\binom{a}{\hat{a_p}}}+\frac{\lambda_b S_{\mu_b}^{m,\hat{a_p}}}{\lambda}+\frac{\lambda_c S_{\mu_c}^{m,\hat{a_p}}}{\lambda}\right)
%   \binom{a}{\hat{a_p}}p^{\hat{a_p}}(1-p)^{a-\hat{a_p}}P(A=a)}
% {\sum\limits_{n\geq a_p} \left(\frac{\lambda_a}{\lambda}\frac{\binom{m-1}{\hat{a_p}-1}}{\binom{n}{\hat{a_p}}}+\frac{\lambda_b S_{\mu_b}^{m,\hat{a_p}}}{\lambda}+\frac{\lambda_c S_{\mu_c}^{m,\hat{a_p}}}{\lambda}\right)
%   \binom{n}{\hat{a_p}}p^{\hat{a_p}}(1-p)^{n-\hat{a_p}}P(A=n)}
% \end{multline}
% \normalsize
\small
\begin{multline}
\nonumber
p_A(a)=\\
\frac{\left(\lambda_a\binom{m-1}{\hat{a_p}-1}+\lambda_b S_{\mu_b}^{m,\hat{a_p}} \binom{a}{\hat{a_p}}+\lambda_c S_{\mu_c}^{m,\hat{a_p}} \binom{a}{\hat{a_p}}\right)
(1-p)^{a}\pi(a,\mu_a)}
{\sum\limits_{n\geq a_p}\left(\lambda_a\binom{m-1}{\hat{a_p}-1}+\lambda_b S_{\mu_b}^{m,\hat{a_p}} \binom{n}{\hat{a_p}}+\lambda_c S_{\mu_c}^{m,\hat{a_p}} \binom{n}{\hat{a_p}} \right)
(1-p)^{n}\pi(n,\mu_a)}
\end{multline}
\normalsize
\end{proof}

\begin{figure}[htbp]
\begin{center}
\includegraphics[width=\linewidth, keepaspectratio]{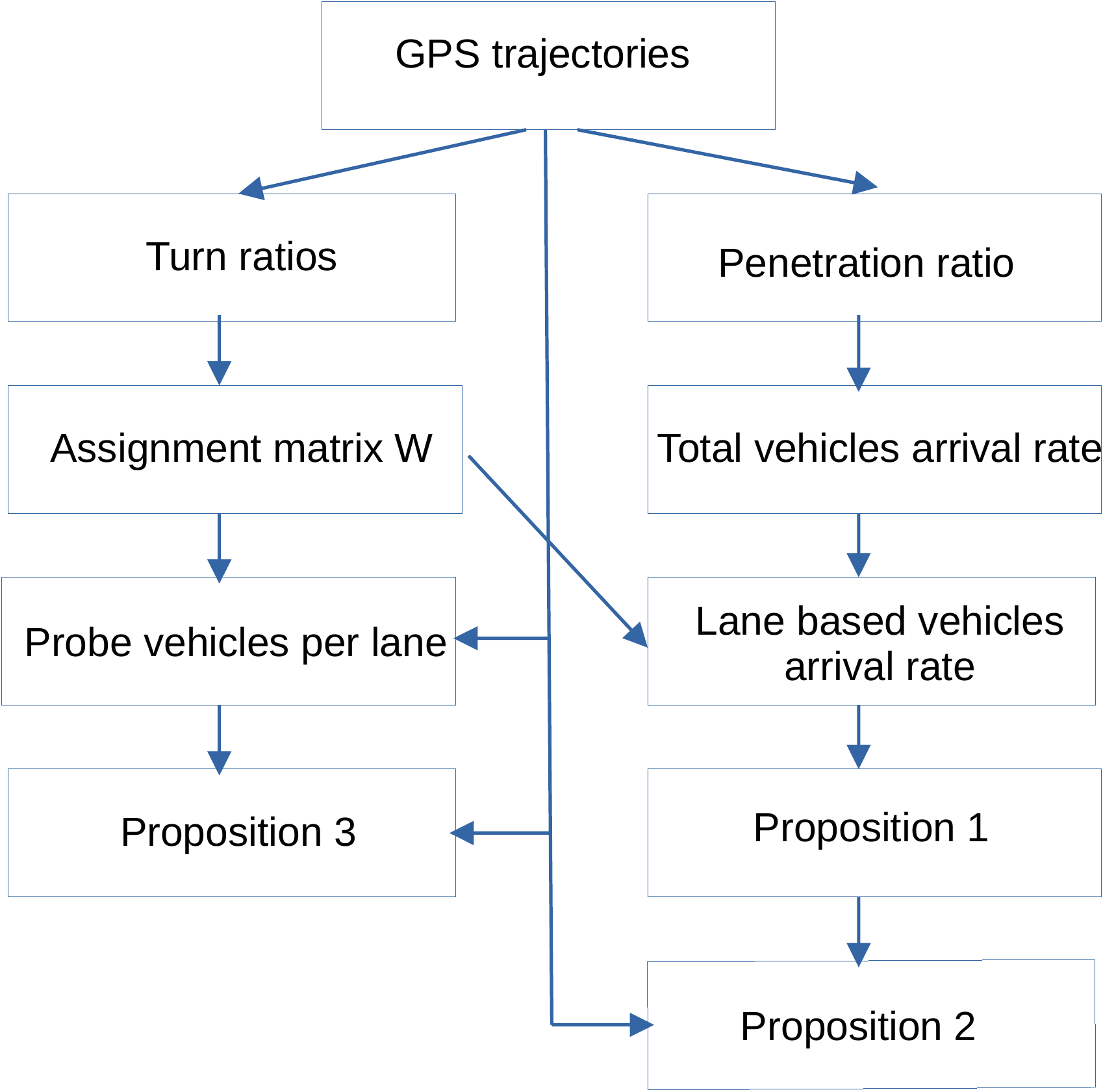}
\caption{Summary of the estimation process}
\label{fig:block_diagram}
\end{center}
\end{figure}
In Fig.\ref{fig:block_diagram}, we have summarized the process of the road traffic estimation.
Based on the trajectories of the probe vehicles, the turn ratios and the penetration ratio of probe vehicles are estimated.
The assignment matrix $W$ is derived as the solution of an optimization problem which tends to equilibrate the queues per lane as much as possible.
The total vehicles arrival rate estimation combined with the assignment matrix $W$ gives the vehicles arrival rate per lane.
On another hand, the number of probe vehicles per lane is estimated with the assignment matrix $W$ and the trajectories of the probe vehicles.
Proposition~\ref{prop1} is derived from the vehicles arrival rate per lane.
Proposition~\ref{prop2} is extending Proposition~\ref{prop1} while using the information given by the probe vehicles.
Finally, Proposition~\ref{prop3} uses the estimated number of probe vehicles per lane combined with the information given by the probe vehicles,
in order to estimate the queue length distribution probability per lane.

\subsection{Application to roads composed of any number of lanes}
\label{sub-n-lanes}
The method we have proposed for the three propositions can be generalized easily to roads composed of an any number of lanes.
However, we give in this section insights into the n-lanes roads case and we demonstrate that the 3-lanes roads case is enough to address the general case.
Indeed, with the accuracy of the GPS localization system nowadays, we can assign a vehicle with 95\% probability within a radius of $7.8 m$, as it is written in~\cite{gps.gov} :
``the government commits to broadcasting the GPS signal in space with a global average user range error (URE) of $\leq7.8 m$ ($25.6$ ft.), with $95\%$ probability.''.

As the standard lane width in the United States is $3.7 m$, we can assign any vehicle located on a lane to a virtual three lanes road.
All the virtual roads composed of successive three lanes are enumerated given the topology of the road junction. 
The virtual roads can overlap such that a lane can be in many virtual roads.
The estimation of a queue length for a lane which is present only in a single virtual 3-lanes road is straightforward with the method we have exposed.
For the other lanes, which are represented in many virtual roads, we will take the average of the estimation done in each of the virtual roads.
By taking the average of the queue length estimations for a given lane, on the set of all the virtual roads, 
we counterbalance the inaccuracy due to the estimation of $A_p=\hat{a_p}$, $B_p=\hat{b_p}$, $C_p=\hat{c_p}$.
Indeed, we recall that we have chosen to take these variables as a function of the expected values : 
$\hat{a_p}=[\mathbb E(A_p)]$, $\hat{b_p}=[\mathbb E(B_p)]$, $\hat{c_p}=[\mathbb E(C_p)]$.

\section{Numerical experiments}
\label{sec:numerical}
\subsection{Estimation of the primary parameters}
\label{subsec:primary_experiments}
In this section we discuss the model with simulation results.
We have implemented the simulation model with Omnet++~\cite{Varga01theomnet++}.
In the Omnet++ implementation, we represent vehicles with packets which arrive as a Poisson arrival process of rate $\lambda$.
The packets (vehicles) are assigned to the queue in accordance with matrix $W$
and their destination, and are extracted from each queue by a traffic light at a given saturation rate during the green time.
In~\cite{9011732}, we have developed an implementation within Veins framework~\cite{sommer2011bidirectionally} which combines a road traffic simulator (SUMO~\cite{SUMO2012}) and a communication simulator (Omnet++). 
We used to combine SUMO microscopic road traffic simulator to represent the dynamics of the vehicles and Omnet++ to simulate the communication between
the vehicles and the infrastructure which is a Road Side Unit (RSU) in our case. 
The difference between the two implementations is the assignment model, which is not the same in SUMO and in our model.
In the Omnet++ implementation, the assignment simulation model corresponds accurately with our theoretical assignment model.
But in the SUMO implementation, the simulation assignment model is the SUMO lane changing model~\cite{10.1007/978-3-319-15024-6_7}.
This is why in the present paper we use a simulation model using Omnet++ exclusively.

We have used a symmetric and an asymmetric scenarios for the simulation.
For both scenarios, there are three outgoing roads.
In both cases, the topology of the road junction is as follows~: if a vehicle turns left, it must be on the left lane; if it turns right, it must be on the right lane;
if it goes straight, it can be on any of the three lanes.
So we have $\tilde{w}_{21}=\tilde{w}_{31}=\tilde{w}_{13}=\tilde{w}_{23}=0$; see section~\ref{sub-primary}.
The matrix W in the two scenarios is the result of the optimization problem given in~(\ref{eq:w2}) with the constraints~(\ref{cont3})
on the variable $\tilde{w}$ given by the junction topology.

\textbf{Scenario S1}
The turn ratios for the symmetric scenario S1 are $(\rho_1, \rho_2, \rho_3)=(0.1,0.8,0.1)$. The total arrival rate $\lambda=0.75$ vehicles/second.
Then the optimization problem~(\ref{eq:w2}) gives the matrix $W$ given in TABLE~\ref{tab:symetric_scenario}.
We observe that this scenario S1 is symmetric, because the incoming lanes are balanced~: $w_{i_a}=w_{i_b}=w_{i_c}=1/3$.
\begin{table}[htbp]
\center
\begin{tabular}{|c|c|c|c|c|}
\hline
Symmetric scenario S1& \multicolumn{3}{c|}{Destination} &\multirow{2}{*}{$w_i$}\\
\cline{1-4}
Incoming lane & 1 & 2 & 3 & \\
\hline
$i_a$ & 0.1 & 0.23 & 0 & 1/3\\
$i_b$ & 0 & 0.33 & 0 & 1/3\\
$i_c$ & 0 & 0.23 & 0.1 & 1/3\\
\hline
$\rho_j$& 0.1 & 0.8 & 0.1&\\
\hline
\end{tabular}
\caption{Matrix $W$ for the symmetric scenario. Total arrival rate=$0.75$ vehicles/second.}
\label{tab:symetric_scenario}
\end{table}

\textbf{Scenario S2}
The turn ratios for the asymmetric scenario S2 are $(\rho_1, \rho_2, \rho_3)=(0.7,0.15,0.15)$. The total arrival rate $\lambda=0.5$ vehicles/second.
Then the optimization problem~(\ref{eq:w2}) gives the matrix $W$ given in TABLE~\ref{tab:asymetric_scenario}.
The scenario S2 is asymmetric since $w_{i_a} \neq w_{i_b}$ and $w_{i_a} \neq w_{i_c}$.

\begin{table}[htbp]
\center
\begin{tabular}{|c|c|c|c|c|}
\hline
Asymmetric scenario S2 & \multicolumn{3}{c|}{Destination} & \multirow{2}{*}{$w_i$}\\
\cline{1-4}
Incoming lane & 1 & 2 & 3 & \\
\hline
$i_a$ & 0.7 & 0 & 0 & 0.7\\
$i_b$  & 0 & 0.15 & 0 & 0.15\\
$i_c$  & 0 & 0 & 0.15 & 0.15\\
\hline
$\rho_j$& 0.7 & 0.15 & 0.15&\\
\hline
\end{tabular}
\caption{Matrix $W$ for the asymmetric scenario. Total arrival rate=$0.5$ vehicles/second.}
\label{tab:asymetric_scenario}
\end{table}
In Fig.~\ref{fig:p} we display the penetration ratio estimator we have proposed in section~\ref{sub-primary}.
The estimator $\hat{p}=\frac{\sum_n X_p(n)}{q_{sat}\sum_j \max_k\{t_e^{k,j}\}}$ generalizes the 
estimators proposed in~\cite{COMERT2016502} to the multi-lanes case.
We observe on Fig.~\ref{fig:p} that for $p>0.15$ the estimation error is very low compared to the case $p<0.15$.

\begin{figure}[htbp]
\begin{center}
\includegraphics[width=0.9\linewidth, keepaspectratio]{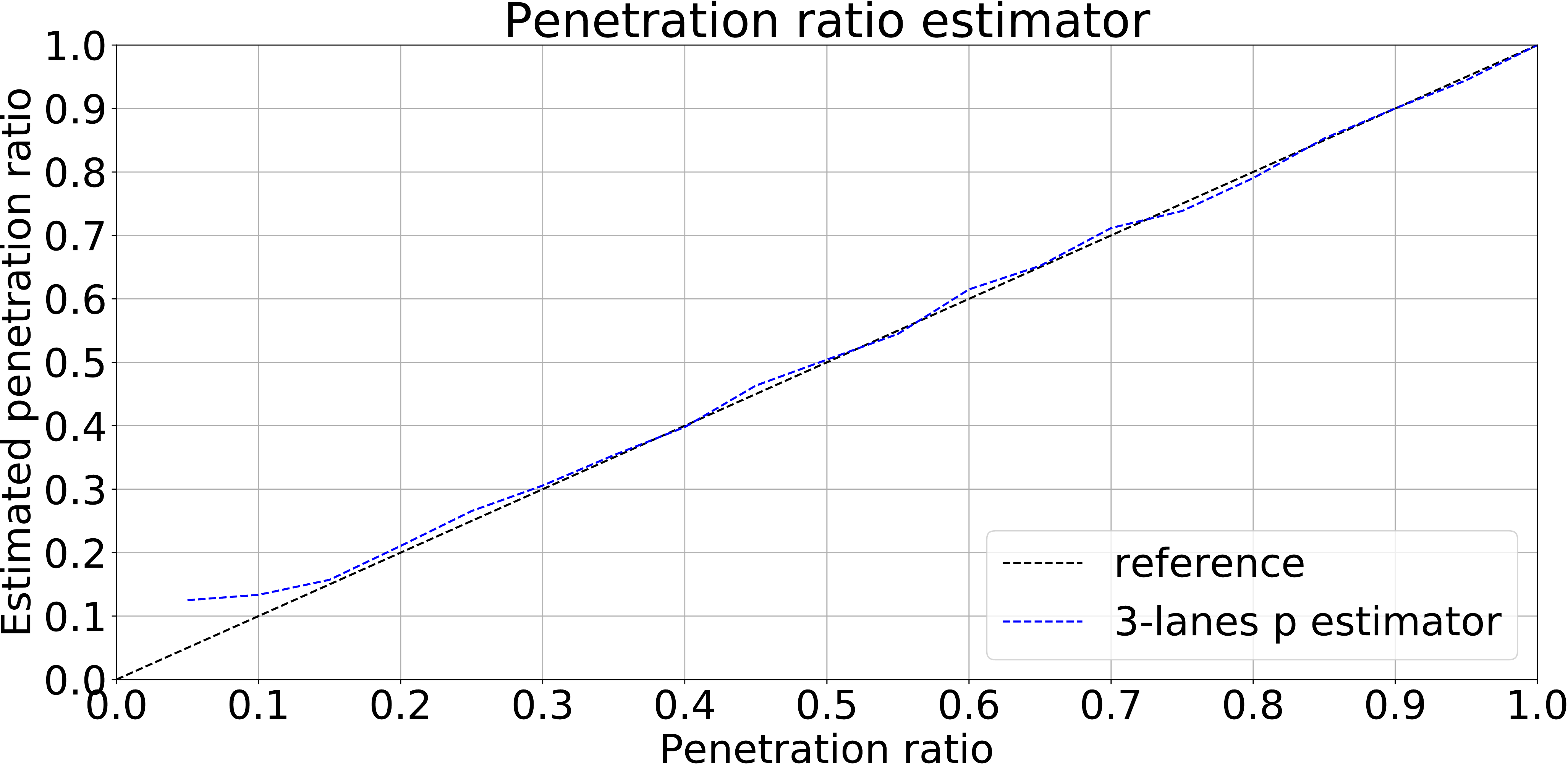}
\caption{Penetration ratio estimator. Simulated time=$25$ hours.}
\label{fig:p}
\end{center}
\end{figure}

In Fig.~\ref{fig:probes}, we show the mean absolute error
of various estimators for computing $A_p$, $B_p$, $C_p$.
The first estimator is the one which estimates the number of probes per incoming lane, based on the destination road of each probe vehicle.
For example, with this estimator we will assume that all the vehicles that have turned right came from the right lane.
We will denote this estimator as $E_0$. With $E_0$, we write :
\begin{equation}
 \hat{a_p}=[\sum_{k}P(D_k=j_1)]
\end{equation}
\begin{equation}
 \hat{b_p}=[\sum_{k}P(D_k=j_2)]
\end{equation}
\begin{equation}
 \hat{c_p}=[\sum_{k}P(D_k=j_3)]
\end{equation}
We propose a more realistic estimator, $E_1$ which is based on Equation~(\ref{eq:probes}) which we recall here :
\begin{equation}
 \hat{a_p}=[\sum_{k} P(O_k=i_a|D_k=j_k)]
\end{equation}
\begin{equation}
 \hat{b_p}=[\sum_{k} P(O_k=i_b|D_k=j_k)]
\end{equation}
\begin{equation}
 \hat{c_p}=[\sum_{k} P(O_k=i_c|D_k=j_k)]
\end{equation}
We notice in Fig.~\ref{fig:probes} that the estimator $E_1$ performs better than the estimator $E_0$.
The estimation error for $E_1$ is at maximum around 1 vehicle, although it can be around 3 vehicles for $E_0$.
This is because with the estimator $E_1$, we use the matrix $W$ to enhance the estimator $E_0$ which estimates
probe vehicles per lane based on the counting of probe vehicles per destination.
We notice also that estimation error for $E_0$ increases linearly with the penetration ratio $p$ while
the estimation error for $E_1$ is clearly sub linear.
\begin{figure}
\centering
 \begin{subfigure}[b]{\linewidth}
  \centering
  \includegraphics[width=\textwidth]{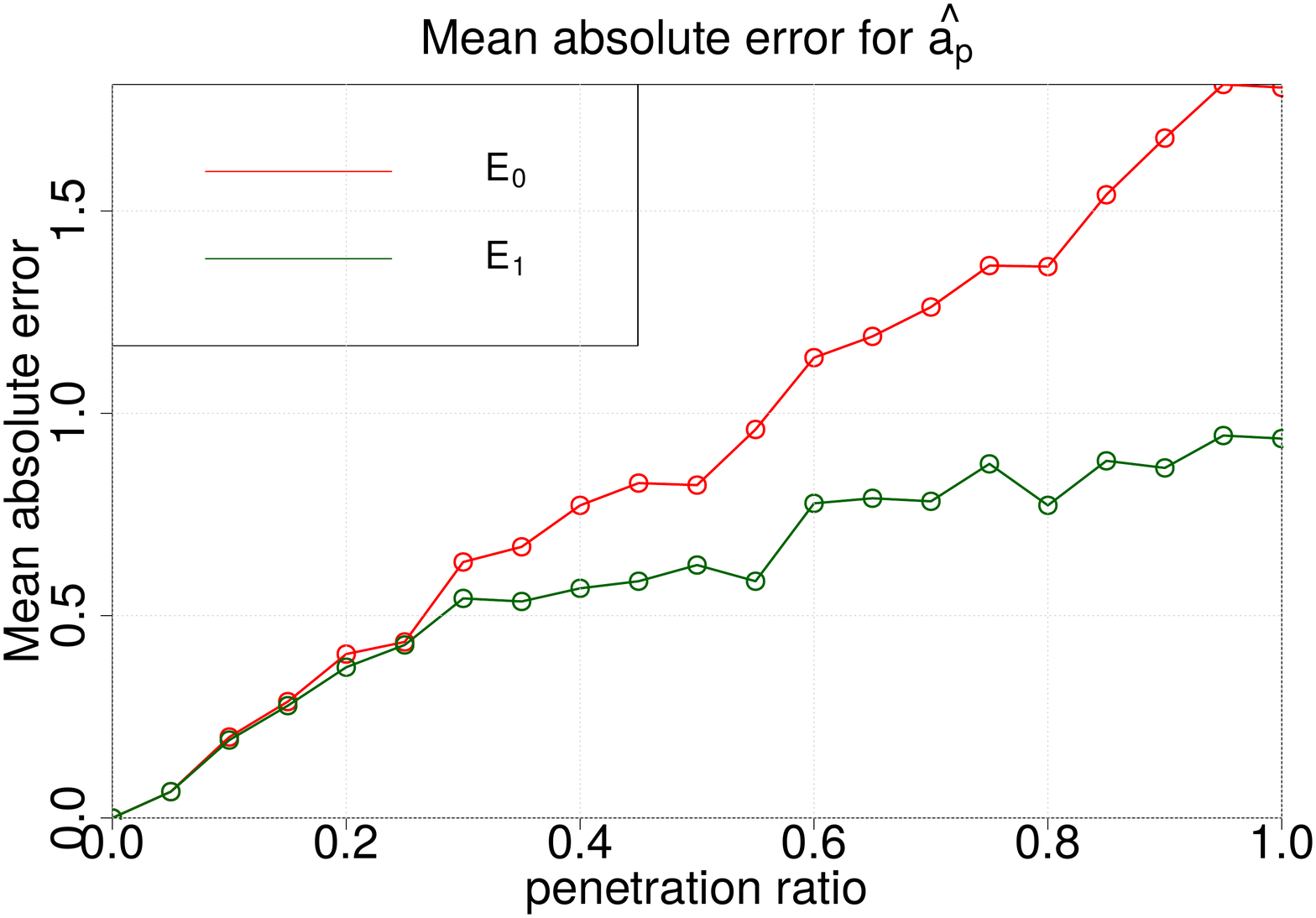}
\end{subfigure}
\begin{subfigure}[b]{\linewidth}
  \centering
  \includegraphics[width=\textwidth]{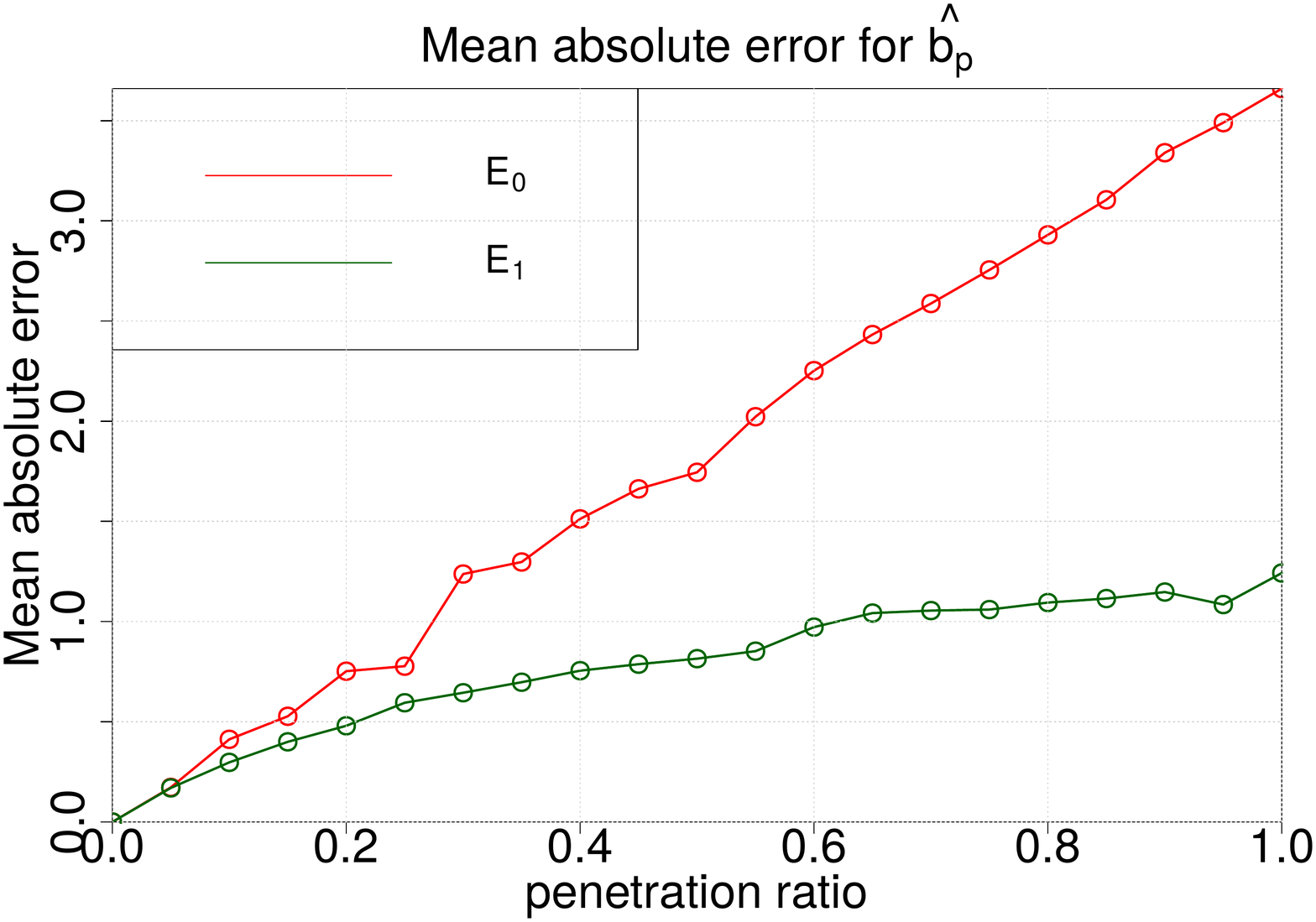}
\end{subfigure}
\begin{subfigure}[b]{\linewidth}
  \centering
  \includegraphics[width=\textwidth]{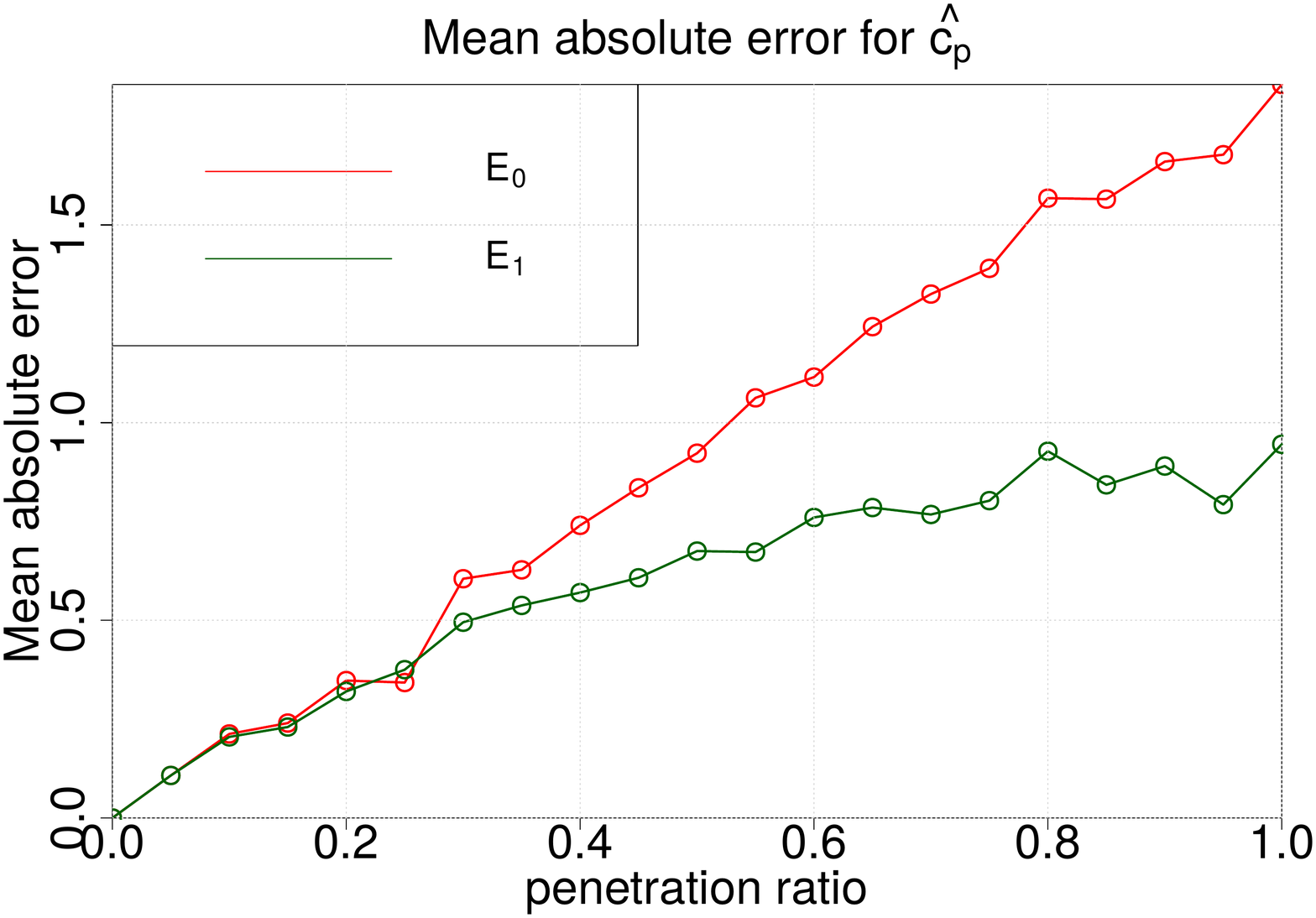}
\end{subfigure}
\caption{Mean absolute error for  $\hat{a_p}$,  $\hat{b_p}$,  $\hat{c_p}$ and estimators $E_0$, $E_1$, averaged over 10 replications. Simulated time = $10$ hours, asymmetric scenario.}
\label{fig:probes}
\end{figure}

\subsection{Road traffic state estimation}
\label{subsec:traffic_experiments}
In this section we evaluate the queue length estimations as given by the expectations of the probability distributions of Propositions~\ref{prop1},\ref{prop2}~and~\ref{prop3}.
We notice that Proposition~\ref{prop2} is extending the work by Comert~\cite{COMERT2009196} to the multi lanes case.

\begin{figure}
\centering
 \begin{subfigure}[b]{\linewidth}
  \centering
  \includegraphics[width=\textwidth]{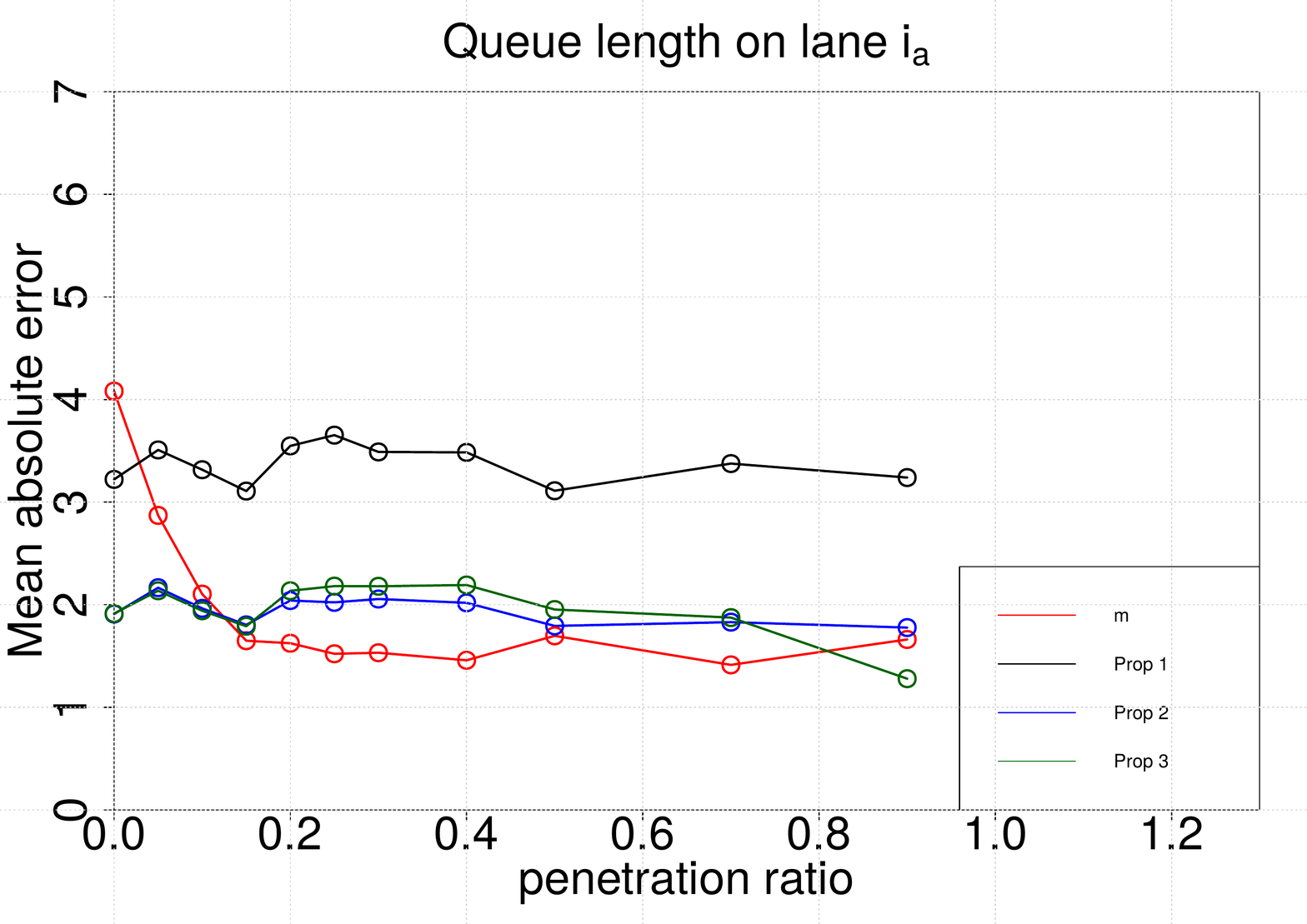}
\end{subfigure}
\begin{subfigure}[b]{\linewidth}
  \centering
  \includegraphics[width=\textwidth]{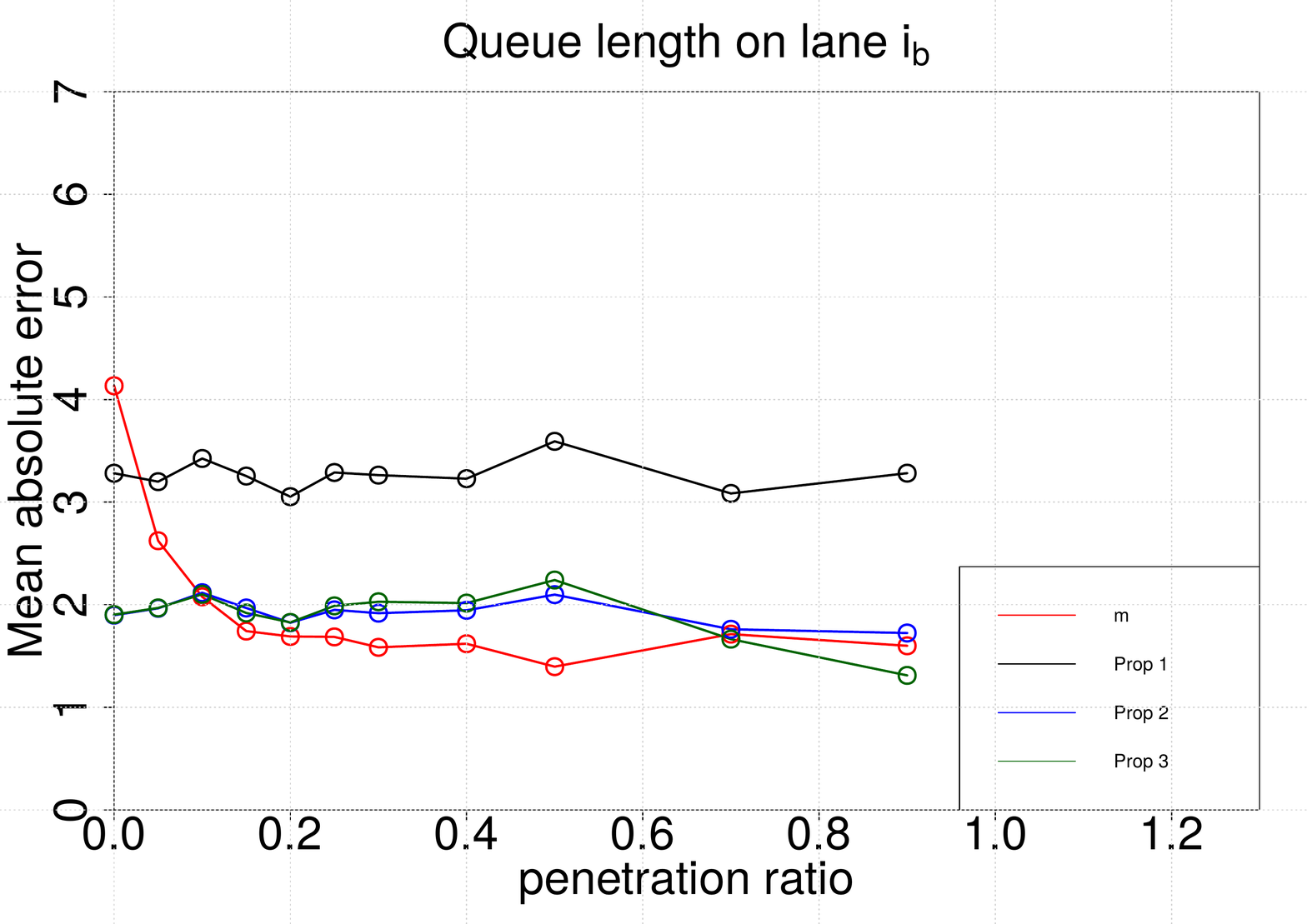}
\end{subfigure}
\begin{subfigure}[b]{\linewidth}
  \centering
  \includegraphics[width=\textwidth]{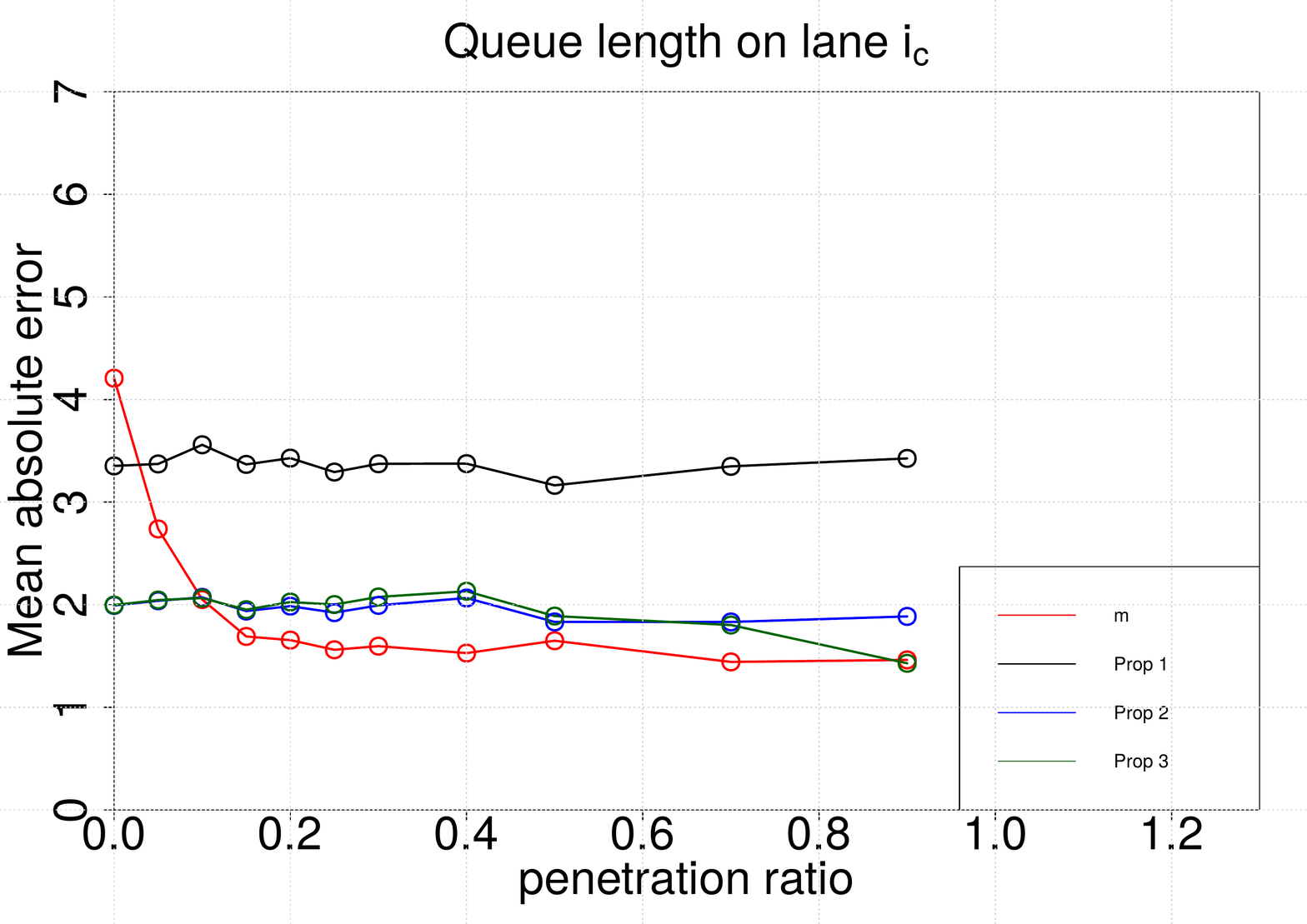}
\end{subfigure}
\caption{Mean absolute error for queue lengths on lanes $i_a$, $i_b$ and $i_c$, with Propositions 1, 2, 3. $m$ means estimating the queue length with $\hat{A}=\hat{B}=\hat{C}=m$ the last probe location. Simulated time = $2.5$ hours and symmetric scenario.}
\label{fig:prop}
\end{figure}

\begin{figure}
\centering
 \begin{subfigure}[b]{\linewidth}
  \centering
  \includegraphics[width=\textwidth]{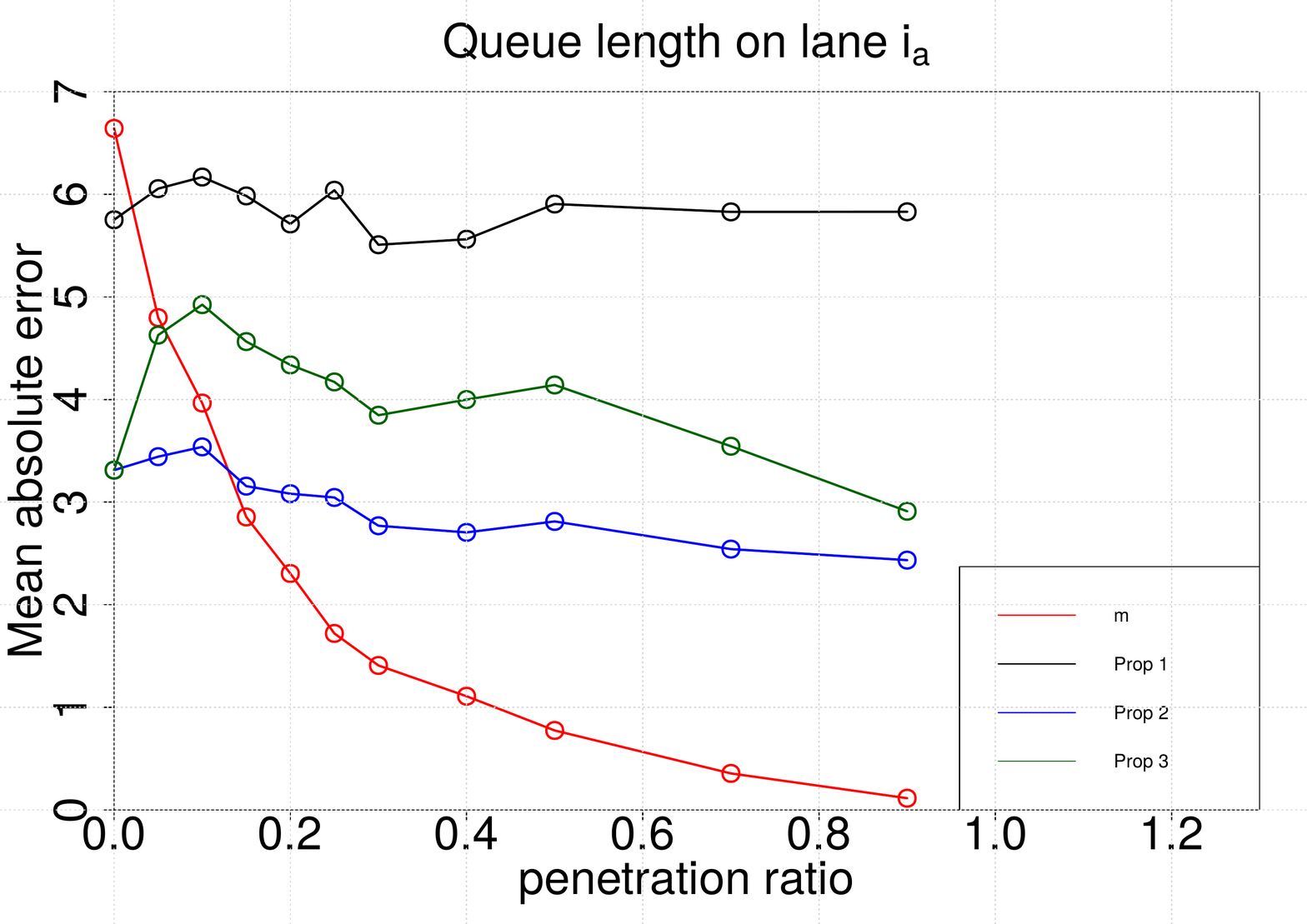}
\end{subfigure}
\begin{subfigure}[b]{\linewidth}
  \centering
  \includegraphics[width=\textwidth]{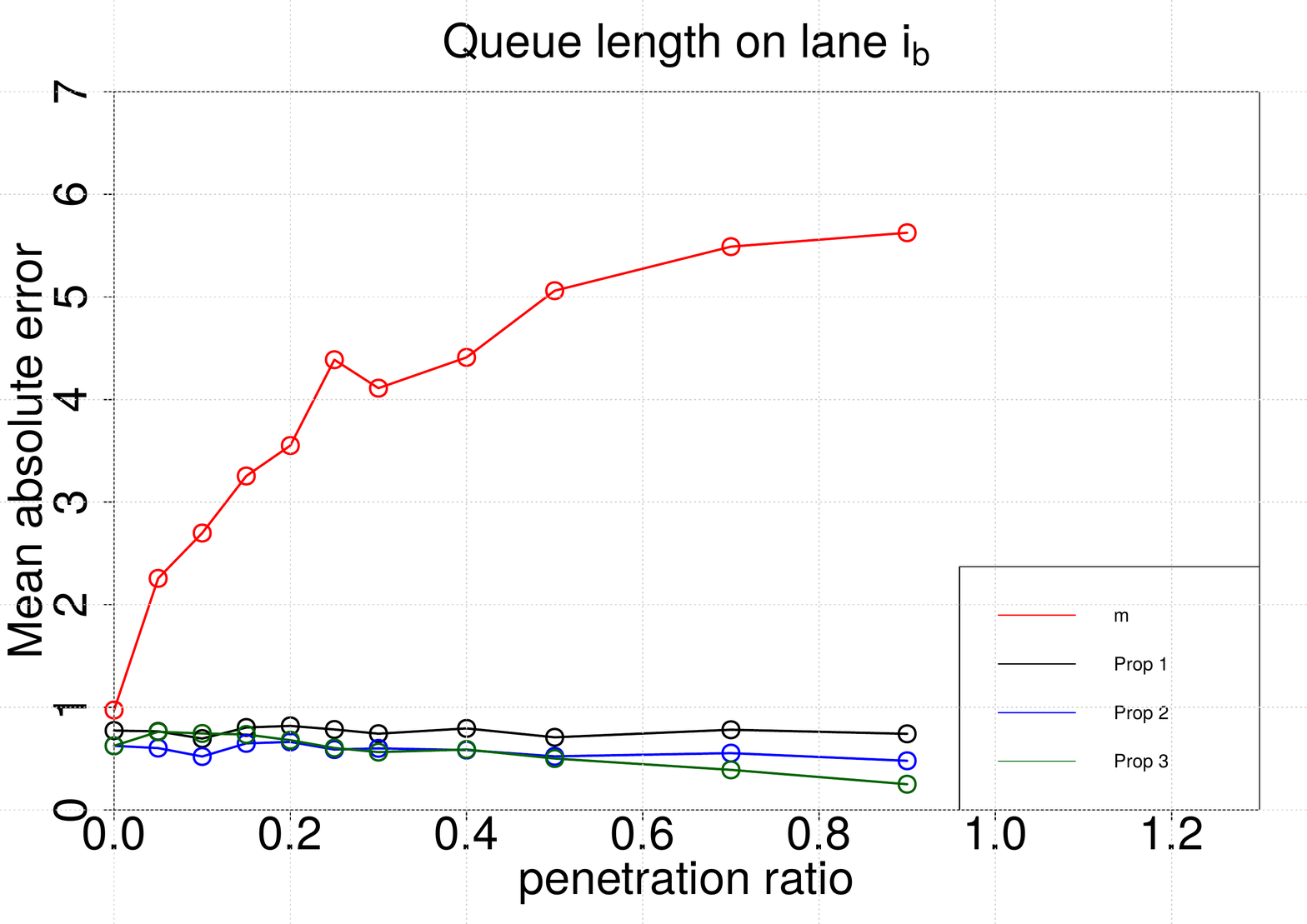}
\end{subfigure}
\begin{subfigure}[b]{\linewidth}
  \centering
  \includegraphics[width=\textwidth]{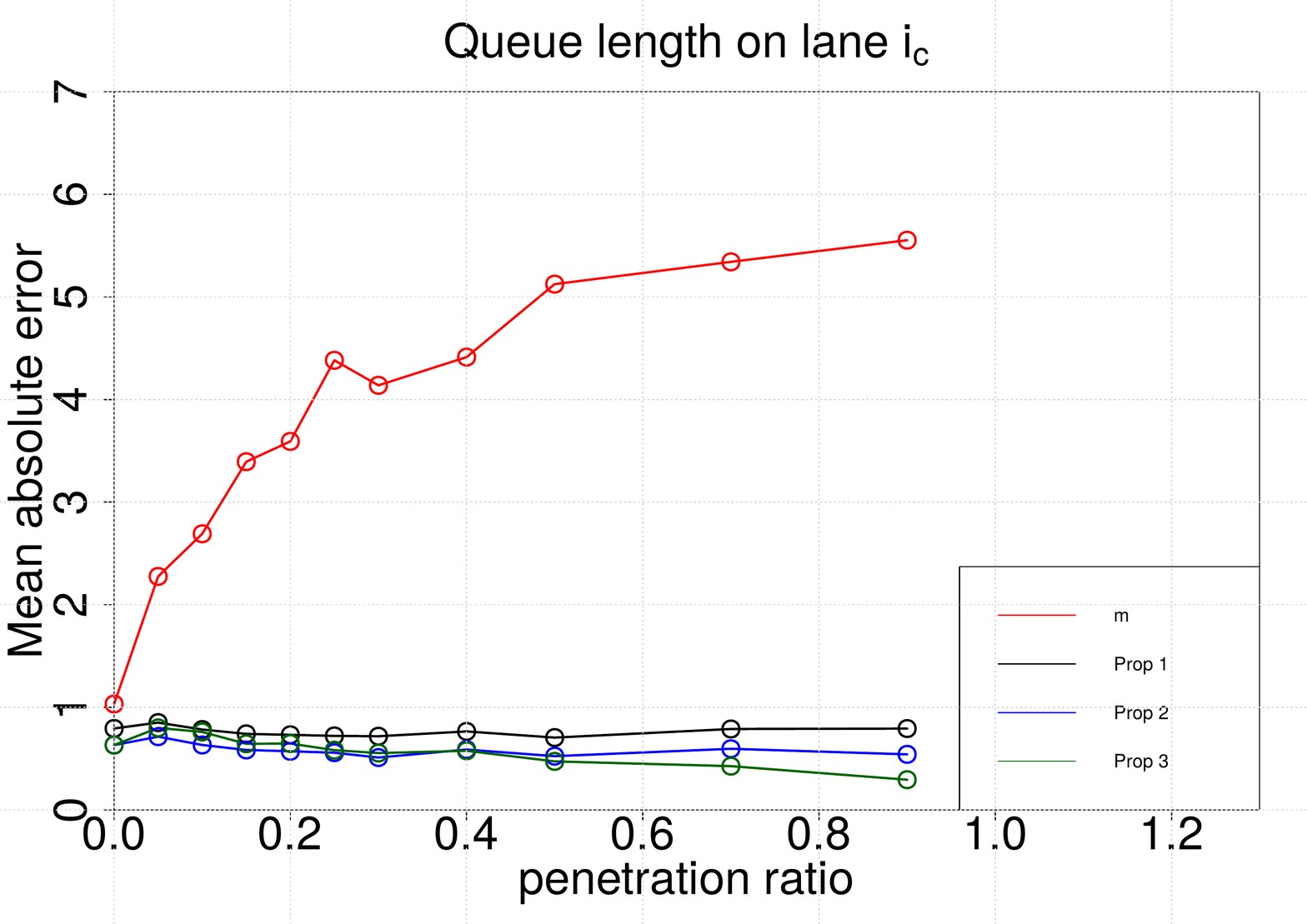}
\end{subfigure}
\caption{Mean absolute error for queue lengths on lanes $i_a$, $i_b$ and $i_c$, with Propositions 1, 2, 3. $m$ means estimating the queue length with $\hat{A}=\hat{B}=\hat{C}=m$ the last probe location. Simulated time = $2.5$ hours and asymmetric scenario.}
\label{fig:prop2}
\end{figure}

In Fig.~\ref{fig:prop} and Fig.~\ref{fig:prop2}, we compare the mean absolute error between the queue lengths and their estimated values with the different propositions respectively 
for scenario S1 and S2.
We notice on Fig.~\ref{fig:prop} that the estimator $m$ performs quite well for high penetration ratios, because we consider here the symmetric scenario.
On the other hand, we notice on Fig.~\ref{fig:prop} that the estimator $m$ is not accurate for low penetration ratios even in this symmetric scenario.
We notice also that Proposition~\ref{prop2} and Proposition~\ref{prop3} perform better than Proposition~\ref{prop1} because we add some information to the estimations.
In addition, we notice that all the estimators increase their accuracy as the penetration ratio of probe vehicles increases.

Let us now look at the simulation results of the asymmetric scenario.
We notice on  Fig.~\ref{fig:prop2} that although the estimator $m$ performs quite well on lane $i_a$, it is not true for lanes $i_b$ and $i_c$ because the demand is asymmetric.
The estimator $m$ is not an option for lanes $i_b$ and $i_c$. 
On another hand, Proposition~\ref{prop2} and Proposition~\ref{prop3} perform better than Proposition~\ref{prop1} for the all the lanes.
Proposition~\ref{prop2} performs better than Proposition~\ref{prop3} for the longest lane $i_a$ although for the lanes $i_b$ and $i_c$ Proposition~\ref{prop3}
is more accurate than Proposition~\ref{prop2}.
%\clearpage
\section{Conclusion and perspectives}
\label{sec-conclusion}
A method to estimate the road traffic state at a multi-lanes road junction has been proposed.
We have given an estimator for the penetration ratio of communicating vehicles as well as the arrival rate of vehicles.
Based on the assumption that the queues tend to balance, we have derived an assignment matrix which gives the probabilities that a vehicle comes from an origin lane
and goes to a destination road.
We have extended an existing method for road traffic state estimation on 2-lanes roads, to the case where roads are composed of any number of lanes.
Three estimators for the queue lengths at the road junction have been proposed.
We have implemented the model with a discrete event simulator where vehicles are represented by packets.
Numerical experiments allow us to discuss the propositions and confirm that the model performs good especially for the asymmetric traffic demand scenarios.
Concerning the future works regarding the present paper, we notice that the road traffic demand was assumed to be moderate or low such that the arrival process can be considered as a Poisson process.
A future work could be to address the case of high traffic demand by taking into consideration the overflow queue.
On another hand, it seems relevant to have assumed that the vehicles tend to choose the lane with the shortest queue although it lacks real assignment data to confirm this assumption.
Finally, even if the GPS localization system becomes more accurate, there will be a need for robust models which can take into account inacurracies on the localization of probe vehicles.

\bibliographystyle{IEEEtran}
\bibliography{./estimation_nlanes}
%\nocite{*}
%\balance
% \begin{IEEEbiography}[{\includegraphics[width=1in,height=1.25in,clip,keepaspectratio]{./img/biosketch/Cyril_Nguyen_Van_Phu_27042019.jpg}}]
% {Cyril Nguyen Van Phu}
% received the Master of Science/Engineer’s degree in 2001 from Institut des Sciences et Technologies of Universit\'e Pierre et Marie Curie, Paris, France.
% In 2002, he joined Laboratoire Central des Ponts et Chauss\'ees, Department of Physics. 
% Since 2015, he joined Grettia laboratory of Univ Gustave Eiffel Ifsttar, where he researches on intelligent transportation systems. 
% His current fields of interest include traffic flow theory and simulation, 
% communication networks, queuing systems and probability theory.  
% \end{IEEEbiography}
% \begin{IEEEbiography}[{\includegraphics[width=1in,height=1.25in,clip,keepaspectratio]{./img/biosketch/Nadir_Farhi.jpg}}]
% {Nadir Farhi}
% received the Engineer degree in operations research from the University of Bejaia, Algeria, in 2003, the M.Sc. and Ph.D. degrees in applied mathematics 
% from Paris 1 University, France, in 2004 and 2008, respectively, and the Habilitation degree in mathematics from University Paris Est, France in 2018. 
% He was worked as a Post-Doctoral Researcher with The University of Texas at Dallas in 2009, with INRIA Paris, and then with ENS Paris in 2010. 
% He joined the Grettia Laboratory, Univ Gustave Eiffel, as a Permanent Researcher, in 2011. His research interests cover traffic flow theory, control, modeling, and optimization of transportation systems.
% \end{IEEEbiography}
\end{document}